\documentclass[preprint]{vldb}
\toappear{}
\makeatletter
\def\@mkbibcitation{}
\makeatother

\usepackage{graphicx}
\usepackage{balance}  % for  \balance command ON LAST PAGE  (only there!)
\usepackage[sort,compress]{cite}
\usepackage{braket}

\usepackage{amsthm}
\newtheorem{theorem}{Theorem}
\newtheorem*{theorem*}{Theorem}
\newtheorem{definition}{Definition}
\newtheorem*{definition*}{Definition}
\newtheorem{lemma}{Lemma}
\newtheorem*{lemma*}{Lemma}

\usepackage[caption=false,font=footnotesize,labelformat=simple]{subfig}
\usepackage{algorithm}
\usepackage[noend]{algpseudocode}

\algrenewcommand\algorithmicindent{1.0em}

\renewcommand{\algorithmicforall}{\textbf{for each}}
\algdef{SxnE}{ParFor}{EndParFor}[1]{\algorithmicfor\ #1\ \textbf{do in parallel}}
\algdef{SxnE}{ParForAll}{EndParForAll}[1]{\algorithmicforall\ #1\ \textbf{do in parallel}}

\makeatletter
\newcommand{\pushcode}[1][1]{\hskip\ALG@thistlm \hskip#1\dimexpr\algorithmicindent\relax}
\algnewcommand{\LineComment}[1]{\Statex \hskip\ALG@thistlm \(\triangleright\) #1}
\renewcommand{\ALG@beginalgorithmic}{\small}
\makeatother

\algnewcommand{\LineCommentIndented}[1]{\Statex \pushcode[2] \(\triangleright\) #1}
\algnewcommand{\LineCommentIndented1}[1]{\Statex \pushcode[1] \(\triangleright\) #1}
\algrenewcommand\alglinenumber[1]{\scriptsize #1\,\,}

\newcommand{\bhline}[1]{\noalign{\hrule height #1}}

\begin{document}
% Title and Authors <<<

% ****************** TITLE ****************************************

\title{Fast Subgraph Matching by Exploiting Search Failures}

% possible, but not really needed or used for PVLDB:
%\subtitle{[Extended Abstract]}

% ****************** AUTHORS **************************************


\numberofauthors{4} %  in this sample file, there are a *total*

\author{
% 1st. and 2nd. author
\alignauthor
Junya Arai\\
       \affaddr{Nippon Telegraph and Telephone Corporation}\\
       \email{junya.arai.at@hco.ntt.co.jp}
\alignauthor
Makoto Onizuka\\
       \affaddr{Osaka University}\\
       \email{onizuka@ist.osaka-u.ac.jp}
\alignauthor
Yasuhiro Fujiwara\\
       \affaddr{Nippon Telegraph and Telephone Corporation}\\
       \email{yasuhiro.fujiwara.kh@hco.ntt.co.jp}
% 3rd. author
\and  % use '\and' if you need 'another row' of author names
\alignauthor
Sotetsu Iwamura\\
       \affaddr{Hokkaido University}\\
       \email{s-iwamura@mcip.hokudai.ac.jp}
}
% There's nothing stopping you putting the seventh, eighth, etc.
% author on the opening page (as the 'third row') but we ask,
% for aesthetic reasons that you place these 'additional authors'
% in the \additional authors block, viz.
%\additionalauthors{Additional authors: John Smith (The Th{\o}rv\"{a}ld Group, {\texttt{jsmith@affiliation.org}}), Julius P.~Kumquat
%(The \raggedright{Kumquat} Consortium, {\small \texttt{jpkumquat@consortium.net}}), and Ahmet Sacan (Drexel University, {\small \texttt{ahmetdevel@gmail.com}})}
\date{Dec 2020}  % TODO: ここどうする？
% Just remember to make sure that the TOTAL number of authors
% is the number that will appear on the first page PLUS the
% number that will appear in the \additionalauthors section.

\maketitle

\renewcommand{\thefootnote}{\fnsymbol{footnote}}
%\footnotetext[0]{This paper is a translation of the following Japanese paper published in March 2018: \begin{CJK}{UTF8}{min}新井淳也, 鬼塚真, 藤原靖宏, 岩村相哲. 探索失敗履歴を用いた高速サブグラフマッチング. 第 10 回データ工学と情報マネジメントに関するフォーラム論文集 (DEIM)\end{CJK}, pp. 1--9, mar 2018, available at https://db-event.jpn.org/deim2018/data/papers/1.pdf.}
\footnotetext[0]{This paper is a translation of the following Japanese paper published in March 2018: Junya Arai, Makoto Onizuka, Yasuhiro Fujiwara, and Sotetsu Iwamura. Fast Subgraph Matching Using a History of Search Failures. Proceedings of the 10th Forum on Data Engineering and Information Management (DEIM), pp. 1--9, mar 2018, available at https://db-event.jpn.org/deim2018/data/papers/1.pdf (In Japanese).}
\renewcommand{\thefootnote}{\arabic{footnote}}
% >>>
\begin{abstract}  % <<<
Subgraph matching is a compute-intensive problem that asks to enumerate all the isomorphic embeddings of a query graph within a data graph.
This problem is generally solved with backtracking, which recursively evolves every possible partial embedding until it becomes an isomorphic embedding or is found unable to become it.
While existing methods reduce the search space by analyzing graph structures before starting the backtracking, it is often ineffective for complex graphs.
In this paper, we propose an efficient algorithm for subgraph matching that performs on-the-fly pruning during the backtracking.
Our main idea is to `learn from failure'.
That is, our algorithm generates failure patterns when a partial embedding is found unable to become an isomorphic embedding.
Then, in the subsequent process of the backtracking, our algorithm prunes partial embeddings matched with a failure pattern.
This pruning does not change the result because failure patterns are designed to represent the conditions that never yield an isomorphic embedding.
Additionally, we introduce an efficient representation of failure patterns for constant-time pattern matching.
The experimental results show that our method improves the performance by up to 10000 times than existing methods.
\end{abstract}
% >>>

\section{Introduction}

Graph data play a central role in the analysis of various information such as the linking structure of the web, social relationships, and financial transactions.
One of the most typical operations on graphs is a subgraph matching query, which enumerates isomorphic embeddings of a query graph within a data graph.
Popular query languages on graph databases, such as SPARQL, Cypher, and Gremlin, provide native support for subgraph matching queries \cite{Voigt2017}.
Subgraph matching is also an essential building block in data analytics applications for investigating human relations \cite{Fang2016}, aiding offline sales \cite{Hu2016}, detecting malware \cite{Park2010}, and so on.
These applications especially utilize graphs whose vertices have a label that represents a type of entity (e.g., person, company, and product).
However, because subgraph matching is NP-hard, its high computational cost often prevents it from practical use \cite{Han2013}.

\begin{figure}[t]
  \centering
  \subfloat[Query graph $Q$]{
    \hspace{30pt}
    \includegraphics[scale=0.4]{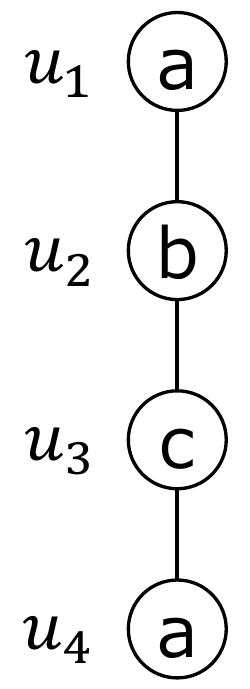}
    \hspace{30pt}
    \label{fig:injection_example_query}
  }
  \subfloat[Data graph $G$]{
    \includegraphics[scale=0.4, page=1]{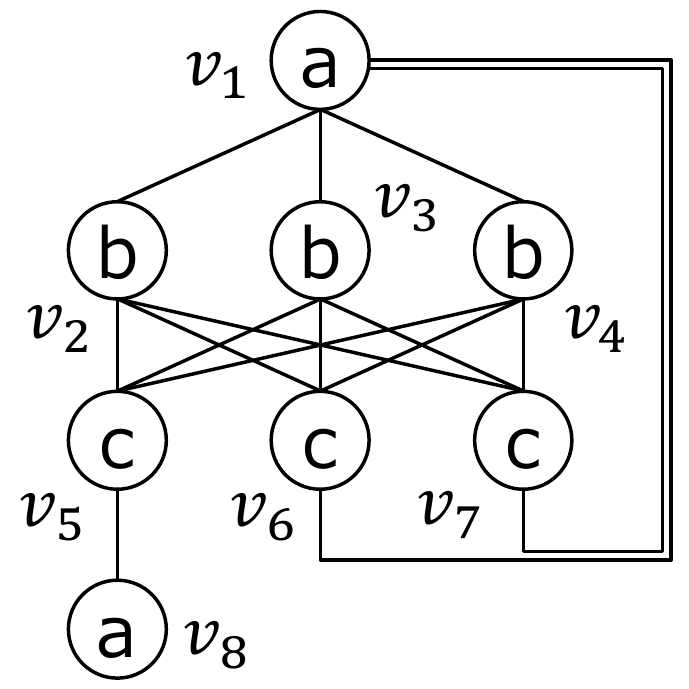}
    \label{fig:injection_example_graph}
  }
  \caption{Example of a query graph and a data graph. $u_i$ and $v_i$ denote a vertex ID, and alphabetic characters in each vertex denote a vertex label}
  \label{fig:injection_example}
\end{figure}

\begin{figure}[t]
  \centering
  \subfloat[Existing methods (e.g., CFL-Match\cite{Bi2016})]{
    \includegraphics[scale=0.39, page=1]{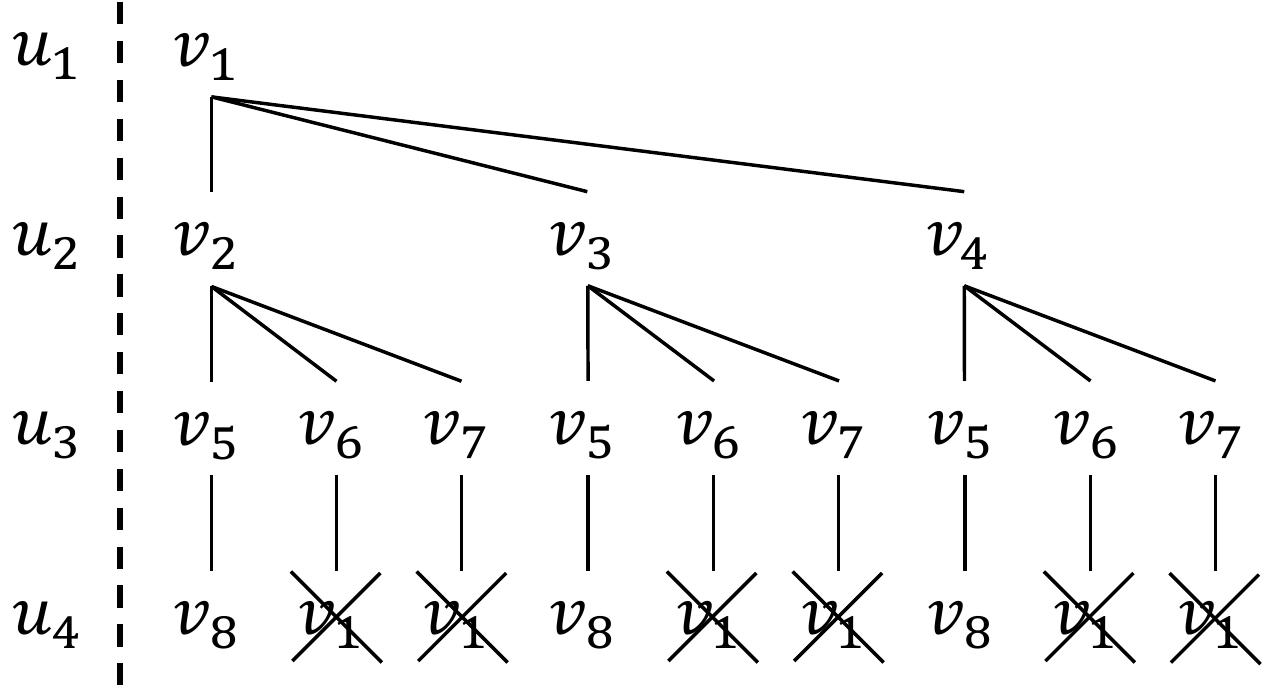}
    \label{fig:injection_example_searchtree_cfl}
  }
  \subfloat[Proposed method]{
    \includegraphics[scale=0.39]{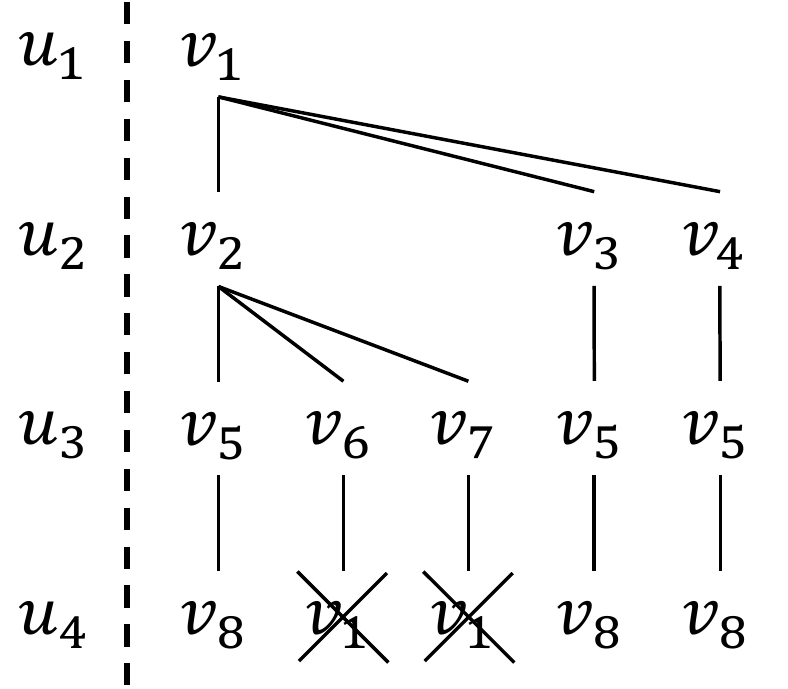}
    \label{fig:injection_example_searchtree_gup}
  }
  \caption{Search trees that show a process of backtracking on $Q$ and $G$ shown in Fig.\ \ref{fig:injection_example}. X-signs denote a search failure.}
  \label{fig:injection_example_searchtree}
\end{figure}

Consequently, many methods have been proposed to improve the performance of subgraph matching \cite{Cordella2004, He2008, Shang2008, Zhao2010, Han2013, Ren2015, Bi2016}.
These methods employ graph structural analysis to reduce the search space of backtracking, which is a general approach for subgraph matching.
The backtracking discovers isomorphic embeddings with a recursive process like depth-first search \cite{Ullmann1976}.
Consider a search tree whose nodes represent a mapping of a single vertex in a query graph (query vertex) onto a vertex in a data graph (data vertex).
The backtracking search starts with an empty partial embedding and adds mappings to it by recursively visiting child nodes.
If the partial embedding becomes an isomorphic embedding, it is reported as a solution.
When the search fails, namely, the partial embedding is found unable to become a solution, it goes back to the state at the parent node and visits another child node.
This process continues until it reports all or a specified number of isomorphic embeddings.
For improving the performance, existing methods reduce the number of search failures by pruning tree nodes before starting the backtracking.
The pruning is based on a comparison of the local structures, such as adjacent labels and spanning trees, between query vertices and data vertices \cite{He2008, Zhao2010, Han2013, Bi2016}.
For example, the state-of-the-art method \cite{Bi2016} produces a search tree shown in Fig.\ \ref{fig:injection_example_searchtree_cfl} for query graph $Q$ and data graph $G$ in Fig.\ \ref{fig:injection_example}.
The backtracking search discovers and reports the first isomorphic embedding by traversing the leftmost path $v_1$-$v_2$-$v_5$-$v_8$.
Denoting a mapping of $u$ onto $v$ by $(u, v)$, this path stands for the search process that incrementally adds $(u_1, v_1)$, $(u_2, v_2)$, $(u_3, v_5)$, and $(u_4, v_8)$ to an empty partial embedding.
Next, the nodes in path $v_1$-$v_2$-$v_6$-$v_1$ are visited.
This path produces a partial embedding that maps both $u_1$ and $u_4$ onto $v_1$.
As we describe in Definition \ref{def:subgraph_isomorphism}, partial embeddings cannot become an isomorphic embedding if they map different query vertices onto the same data vertex.
Thus, the search fails and moves to other tree nodes.
In this way, existing methods perform the backtracking on the pruned search tree.

However, the structural analysis-based pruning is ineffective for complex graphs where the same local structures appear frequently.
For example, while query vertex $u_3$ and data vertices $v_5$, $v_6$, and $v_7$ are adjacent to vertices of label $b$ and $a$ in common, we can obtain an isomorphic embedding only when it maps $u_3$ onto $v_5$, as shown in Fig.\ \ref{fig:injection_example_searchtree_cfl}.
This is because $v_6$ and $v_7$ lack an adjacent vertex which $u_4$ can be mapped onto.
Nevertheless, existing methods \cite{He2008,Zhao2010,Han2013,Bi2016} cannot prune the nodes for $v_6$ and $v_7$ in the search tree.
Since these unpruned nodes repeatedly appear under each of the $u_2$-mapping nodes (i.e., $v_2$, $v_3$, and $v_4$) and repeatedly cause search failures, they significantly increase the number of search failures.
Thus, the performance of existing methods sensitively depends on the structure of given graphs.

Based on these observations, we propose a novel subgraph matching algorithm that `learns from failure'.
In contrast to structural analyses prior to the backtracking, our method performs on-the-fly pruning during the backtracking.
When a partial embedding caused a search failure, our method extracts and records a pattern consisting of the vertex mappings that never appear in isomorphic embeddings.
In other words, a partial embedding always causes a search failure if it matches (i.e., contains) the extracted pattern.
Thus, our method prunes such partial embeddings in the subsequent process of the backtracking.
For example, our method produces a search tree shown in Fig.\ \ref{fig:injection_example_searchtree_gup} for $Q$ and $G$ in Fig.\ \ref{fig:injection_example}.
Let us focus on path $v_1$-$v_2$-$v_6$-$v_1$ in the tree.
It causes a search failure since both $u_1$ and $u_4$ is mapped onto $v_1$.
In this case, $v_1$ is only one adjacent vertex of $v_6$ with label $a$, and hence a partial embedding inevitably causes a search failure if it contains mappings $(u_1, v_1)$ and $(u_3, v_6)$.
To exploit this property, our method records $\{(u_1, v_1), (u_3, v_6)\}$ as a pattern of failing partial embeddings.
Similarly, it records pattern $\{(u_1, v_1), (u_3, v_7)\}$ when the search failed by mapping $u_3$ onto $v_7$ and $u_4$ onto $v_1$.
In the subsequent process, our method prunes partial embeddings if they match these patterns.
As a result, it involves fewer search failures compared with that of existing methods, as shown in Fig.\ \ref{fig:injection_example_searchtree}.

Our method has three advantages:\\
(1) \textbf{Robust:}
While existing methods focus on structural analyses before the backtracking, our method exploits information offered by the backtracking.
This makes our method less sensitive to graph structures, and thus it can eliminate search failures that existing methods cannot.\\
(2) \textbf{Scalable:}
Larger query graphs and data graphs make subgraph matching drastically harder.
This is because it suffers from exponential growth of the number of possible mappings between query vertices and data vertices.
Thanks to the effective pruning, our method can handle larger-scale graphs within practical time.\\
(3) \textbf{Exact:}
Some existing methods perform efficient subgraph matching at the sacrifice of the exactness.
In contrast to them, our method is proved to prune only unnecessary partial embeddings and hence exactly enumerates all the isomorphic embeddings.

The experimental results show that our method is up to 10000 times faster than existing methods.
With those methods, complex query sets are almost intractable because they often need more than one day to respond.
On the other hand, our method can respond to them within minutes or seconds on average.

This paper is organized as follows.
Section 2 reviews related work, and Section 3 gives the preliminaries.
Sections 4 and 5 present details of the proposed method and the experimental results, respectively.
Finally, Section 6 concludes this paper.

\section{Related Work}
\label{sec:related_work}

We can find two kinds of problem settings among the studies on subgraph matching.

The first one enumerates subgraphs isomorphic to a query graph within a single, large data graph.
This paper focuses on this problem setting.
For this purpose, Ullmann \cite{Ullmann1976} is one of the most traditional and well-known algorithms.
It originated subgraph matching based on backtracking.
After Ullmann, many studies have been conducted mostly on two techniques: candidate filtering and matching order selection.
Candidate filtering produces a candidate set $C[u_i]$ for each query vertex $u_i$.
$C[u_i]$ is a set of data vertices which $u_i$ can be mapped onto.
This technique has been improved to eliminate unnecessary vertices in the candidate set.
Ullmann \cite{Ullmann1976} employs a filter based on vertex labels and degrees.
GraphQL \cite{He2008} and SPath \cite{Zhao2010} add an approximate isomorphism test on local structures around query vertices and data vertices.
The second technique, matching order selection, reorders query vertices so that the backtracking generates fewer partial embeddings.
Because the backtracking generates only valid partial embeddings, they decrease if matching starts from query vertices that have fewer valid partial embeddings.
VF2 \cite{Cordella2004} uses a matching order such that, for arbitrary $i$, query vertices $u_1, u_2, \ldots, u_i$ induce one connected component.
With this ordering, the backtracking can ignore candidate vertices that is not adjacent to the partial embedding generated in the previous step.
In addition to the connectivity, QuickSI \cite{Shang2008} considers rarity of labels to start matching from a substructure of a query graph that is rare in a data graph.
Recent methods perform candidate filtering and matching order selection simultaneously.
TurboISO \cite{Han2013} and CFL-Match \cite{Bi2016} filter candidates using an approximate isomorphism test between a data graph and a spanning tree of a query graph.
Then, they estimate the number of partial embeddings by capturing connectivity between candidate vertices and start matching from a substructure with fewer partial embeddings.
The existing methods above focus on structural analyses before the backtracking.
In contrast to them, our method prunes partial embeddings on-the-fly during the backtracking.
Note that we can also combine our method and structural analyses to obtain higher performance.

The second problem setting is one-versus-many; it takes one query graph and many small data graphs and finds data graphs that contain the query graph.
For this purpose, `filter-and-verify' is a common approach.
This approach builds an index that summarizes each data graph in advance.
For a given query, the `filter` step uses the index to extract data graphs that may contain the query graph.
Then, the `verify` step uses an exact isomorphism test to check that those graphs contain the query graph.
In this problem setting, most studies focus on how to summarize a data graph because it determines accuracy of the `filter` step.
For example, they captures paths \cite{Bonnici2010, Giugno2013} and frequent subgraphs \cite{Cheng2007}.
We can use our method in this problem setting to improve the performance of the `verify' step.

\section{Preliminaries}

This section gives a problem definition, terminology, and a baseline algorithm discussed in the following sections.

\subsection{Problem Definition}

This paper focuses on vertex-labeled undirected graph $G = \left(V_G, E_G, \Sigma, \ell\right)$.
$V_G$ is a set of vertices, $E_G \subseteq V_G \times V_G$ is a set of edges, $\Sigma$ is a set of labels, $\ell$ is a function that maps a vertex to its label.
In subgraph matching, $G$ is called a \emph{data graph}.
This paper also considers \emph{query graph} $Q = \left(V_Q, E_Q, \Sigma, \ell\right)$ whose vertices are numbered like $u_1, u_2, \ldots, u_n$.

\begin{definition}[Subgraph isomorphism]
  \label{def:subgraph_isomorphism}
  $Q$ is \emph{subgraph isomorphic} to $G$ if we can define an embedding $M: V_Q \rightarrow V_G$ that satisfies the following constraints\footnote{We named each constraint for convenience in this paper. In general, these constraints do not have well-accepted names.}:\\
  (1) \textbf{Label constraint:} $\forall u_i \in V_Q,\, \ell(u_i) = \ell(M[u_i])$,\\
  (2) \textbf{Edge constraint:} $\forall (u_i, u'_i) \in E_Q,\, (M[u_i], M[u'_i]) \in E_G$,\\
  (3) \textbf{Injection constraint:}\\
  \hspace{1.8em}$\forall u_i, u'_i \in V_Q,\, u_i \neq u'_i \Rightarrow M[u_i] \neq M[u'_i]$.
\end{definition}

\begin{definition}[Subgraph matching]
Given query graph $Q$ and data graph $G$, \emph{subgraph matching} is a problem to enumerate all the subgraph isomorphic embeddings of $Q$ in $G$.
\end{definition}

While subgraph matching requires to enumerate all the embedding, it is often impractical because the number of embeddings may cause the combinatorial explosion.
Hence, it is a common practice to stop enumerating if the number of found embeddings reaches a specific threshold (e.g., 1000) \cite{Lee2012a, Han2013, Bi2016}.

\subsection{Terminology}

{\setlength{\tabcolsep}{1mm}
\begin{table}[!t]
  \centering
  \caption{Symbols used in this paper}
  \label{tb:notations}
  \begin{tabular}{cl}
    \bhline{1pt}
    \multicolumn{1}{c}{Symbol} & \multicolumn{1}{c}{Definition} \\
    \hline
    $Q, G$                 & Query graph and data graph \\
    $V_Q, V_G$             & Set of vertices \\
    $E_Q, E_G$             & Set of edges \\
    $u_i, v$               & Query vertex and data vertex \\
    $C[u_i]$               & Set of candidate vertices for $u_i$ \\
    $M, \hat{M}$           & Complete embedding and partial embedding \\
    $D, \Gamma$            & Dead-end pattern and dead-end mask \\
    $N(\cdot)$             & Set of neighboring vertices \\
    $\ell(\cdot)$          & Label of the vertex  \\
    $\mathrm{dom}(\cdot), \mathrm{ran}(\cdot)$  & Domain and range of the embedding \\
    $\mathrm{Dead}(\cdot)$ & Predicate: the embedding is a dead-end pattern \\
    \hline
  \end{tabular}
\end{table}
}  % \setlength

Table \ref{tb:notations} lists important symbols used in this paper.
We additionally define some notations and terms.

\begin{definition}[Representation of embeddings]
This paper considers $M$ a set of pairs of query vertex $u_i$ and data vertex $v$.
Specifically, $(u_i, v) \in M$ if and only if $M$ maps $u_i$ onto $v$ (i.e., $v = M[u_i]$).
\end{definition}

\begin{definition}[Domain and range of embeddings]
A domain and a range of embedding $M$ are defined as follows, respectively: $\mathrm{dom}(M) = \Set{u_i | (u_i, v) \in M}$ and $\mathrm{ran}(M) = \Set{v | (u_i, v) \in M}$.
\end{definition}

For example, $M = \{(u_1, v_1),\allowbreak (u_2, v_2),\allowbreak (u_3, v_5),\allowbreak (u_4, v_8)\}$ is an embedding of $Q$ in $G$ shown in Fig.\ \ref{fig:injection_example}.
We also have $\mathrm{dom}(M) = \{u_1,\allowbreak u_2,\allowbreak u_3,\allowbreak u_4\}$ and $\mathrm{ran}(M) = \{v_1,\allowbreak v_2,\allowbreak v_5,\allowbreak v_8\}$.

\begin{definition}[Complete embedding]
We say that $M$ is a \emph{complete embedding} if $M$ is an isomoprhic embedding and $\mathrm{dom}(M) = V_Q$.
\end{definition}

We also use $\hat{M}$ for an embedding and call it a \emph{partial embedding} to emphasize that this embedding may not be a complete embedding.
The reader should be aware that a complete embedding is a special case of a partial embedding.

\subsection{Backtracking in Subgraph Matching}

Algorithm \ref{alg:backtracking} shows a simple backtracking algorithm for subgraph matching.
Recursive function \textproc{Search} takes partial embedding $\hat{M}$ and candidate set $C$.
The simplest way to obtain $C$ is to extract data vertices that have the same label as a query vertex as follows\cite{Ullmann1976}:
\begin{equation}
  \label{eq:label_filter}
  C[u_i] = \Set{ v \in V_G | \ell(v) = \ell(u_i) }.
\end{equation}
A function call $\textproc{Search}(\emptyset, C)$ starts the search.
This function reports $\hat{M}$ if it is a complete embedding and returns (lines 2--5).
Otherwise, the function removes candidates in $C$ that do not satisfy the edge constraint (line 6).
The edge constraint requires data vertices $M[u_i]$ and $M[u'_i]$ to be adjacent if query vertex $u_i$ and $u'_i$ are adjacent.
In other words, we must have $M[u_i] \in N(M[u'_i])$ if $u'_i \in N(u_i)$ and $u'_i \in \mathrm{dom}(\hat{M})$.
Hence, $C$ is refined as follows for each query vertex $u_i$:
\begin{equation}
  \label{eq:candidate_refinement}
  C'[u_i] = C[u_i] \cap \bigcap{}_{u'_i \in N(u_i) \cap \mathrm{dom}(\hat{M})} N(\hat{M}[u'_i]).
\end{equation}
Since $\hat{M}$ contains mappings of $u_1, u_2, \ldots, u_k$, the next step maps $u_{k+1}$ onto one of $v \in C'[u_{k+1}]$ (line 7).
If $v$ is not used in $\hat{M}$ (line 8), it maps $u_{k+1}$ onto $v$ in extended partial embedding $(\hat{M} \cup \{(u_{k+1}, v)\})$ and recurses with it and $C'$ (line 9).
Thus, the backtracking enumerates all the complete embeddings by mapping each query vertex so that the partial embedding satisfies the label constraint (Eq.\ \ref{eq:label_filter}), the edge constraint (line 6), and the injection constraint (line 8).

\begin{algorithm}[!t]
\caption{Backtracking Search}
\begin{algorithmic}[1]
  \Require Query graph $Q$ and data graph $G$
  \Ensure All the embeddings of $Q$ in $G$

  \Function{Search}{$\hat{M}, C$}
    \State $k \leftarrow |\hat{M}|$
    \If{$k = |V_Q|$}
      \State Report $\hat{M}$ as a complete embedding
      \State \Return
    \EndIf
    \State $C' \leftarrow$ Candidates refined with the edge constraints (Eq.\ \ref{eq:candidate_refinement})
    \ForAll{$v \in C'[u_{k+1}]$}
      \If{$v \not\in \mathrm{ran}(\hat{M})$}
	\State $\textproc{Search}(\hat{M} \cup \{(u_{k+1}, v)\}, C')$ \label{ln:hat_m}
      \EndIf
    \EndFor
  \EndFunction
\end{algorithmic}
\label{alg:backtracking}
\end{algorithm}

\section{Method}
\label{sec:method}

This section first presents our main idea for pruning and then details of our proposal.

\subsection{Main Idea}

Our idea is to learn from search failures that occurred during backtracking and to avoid repeating the same failures.
A search failure occurs when a partial embedding contains a vertex mapping that violates any of the three constraints shown in Definition \ref{def:subgraph_isomorphism}.
As long as a partial embedding contains a mapping that causes a search failure, it causes a search failure no matter how the other mappings are changed.
Thus, our method extracts a pattern (i.e., a set of mappings) that violates the constraints every time it encounters a search failure.
The subsequent search prunes partial embeddings that match the extracted patterns.
By avoiding search failures that preprocesses cannot find, this method improves the performance of subgraph matching.

\subsection{Pruning with Dead-end Patterns}
\label{sec:pruning_by_deadend}

This section details our pruning method, referencing the naive backtracking-based search (Algorithm \ref{alg:backtracking}).
First, we define the following terms.
\begin{definition}[Dead-end pattern]
  \label{def:deadend}
  Let $D$ be a partial embedding.
  $D$ is a \emph{dead-end pattern}, or simply a \emph{dead-end}, if a complete embedding $M$ such that $D \subseteq M$ does not exist.
  We also use predicate $\mathrm{Dead}(D)$, which is true if and only if $D$ is a dead-end.
\end{definition}
`Dead-end' is a metaphor of paths in the backtracking search tree that never yield complete embeddings.
To prune unnecessary partial embeddings using dead-end patterns, our method adds the following two procedures in function \textproc{Search}: (i) extracting a dead-end pattern at the end of the function and (ii) pruning a partial embedding that matches any of the dead-end patterns.
The first procedure extracts a subset of mappings in partial embedding $\hat{M}$ as dead-end pattern $D$ if the recursive search starting from $\hat{M}$ fails to find complete embeddings.
$D$ is added to the set of dead-end patterns $\mathcal{D}$.
We detail how to extract dead-end patterns in the next section.
The second procedure first checks if newly generated partial embedding $\hat{M} \cup \{(u_{k+1}, v)\}$ (line \ref{ln:hat_m} of Algorithm \ref{alg:backtracking}) matches any of $D \in \mathcal{D}$.
\begin{definition}[Matching with dead-end patterns]
  Let $\hat{M}$ be a partial embedding and $D$ be a dead-end pattern.
  $\hat{M}$ matches $D$ if $D \subseteq \hat{M}$.
\end{definition}
It skips the recursive call of \textproc{Search} if the partial embedding matches.
Thus, the dead-end pattern extraction and the pruning differentiate our method from the naive backtracking.

\subsection{Dead-end Pattern Extraction}

By Definition \ref{def:deadend}, $\hat{M}$ is a dead-end pattern as-is if it is a subject of the dead-end pattern extraction.
However, since the backtracking generates different partial embeddings for each time, $\hat{M}$ as a dead-end pattern never matches any partial embeddings in the subsequent process of the backtracking.
To obtain dead-end patterns that match many partial embeddings, we need to design rules for extracting a small subset of the vertex mappings from partial embeddings.
Extraction rules give a dead-end mask, defined as follows:

\begin{definition}[Dead-end mask]
\label{def:deadend_mask}
Let $\hat{M}$ be a partial embedding and be a dead-end.
A set of query vertices $\Gamma$ is a \emph{dead-end mask} of $\hat{M}$ if it holds the following formula:
\begin{equation}
\label{eq:pivot_set}
\Gamma \subseteq \mathrm{dom}(\hat{M}) \ \wedge\ \mathrm{Dead}\left(\Set{(u_i, v) \in \hat{M} | u_i \in \Gamma} \right).
\end{equation}
\end{definition}
By using dead-end mask $\Gamma$ of $\hat{M}$, we can extract dead-end pattern $D = \Set{(u_i, v) \in \hat{M} | u_i \in \Gamma}$.

The dead-end masks differ depending on the reason for search failures.
In Algorithm \ref{alg:backtracking}, we can find the following three reasons:
(i) $C'[u_{k+1}]$ is empty (line 7), (ii) all the candidates make $\hat{M} \cup \{(u_{k+1}, v)\}$ non-injective (line 8), and (iii) the recursive calls of \textproc{Search} fail to find a complete embedding (line 9).
In addition, our method prunes the partial embedding if it matches a dead-end pattern.
Thus, it has four reasons for search failures.
For each reason, the following subsections detail rules that give a dead-end mask.

\subsubsection{Case 1: Empty Candidate Set}

If the candidate refinement (line 7) returned an empty candidate set for some query vertices, the search inevitably fails in the subsequent processes of the backtracking.
The following lemma gives a dead-end mask for this case.

\begin{lemma}[Dead-end mask for the `empty candidate set' case]
  \label{lem:pivot_no_candidate}
  Let $u_i$ be a query vertex.
  $\Gamma = N(u_i) \cap \mathrm{dom}(\hat{M})$ is a dead-end mask of $\hat{M}$ if $C'[u_i] = \emptyset$.
\end{lemma}
\begin{proof}
Let us assume $\Gamma$ is not a dead-end mask.
Then, $\Gamma$ does not hold Eq. \ref{eq:pivot_set}.
Since we have $\Gamma \subseteq \mathrm{dom}(\hat{M})$ from its definition, and hence there is complete embedding $M$ such that $\{\, (u_j, v) \in \hat{M} \,|\, u_j \in \Gamma \} \subseteq M$.
By substituting $M$ to $\hat{M}$ in Eq.\ \ref{eq:candidate_refinement}, we obtain $C'_M[u_i] = C[u_i] \cap \bigcap{}_{u'_i \in N(u_i)} N(M[u'_i])$.
Since $M$ satisfies the edge constraint, $M[u_i] \in C'_M[u_i]$.
However, we have $C'_M[u_i] \,\subseteq\, C[u_i] \cap \bigcap{}_{u'_i \in \Gamma} N(\hat{M}[u'_i]) = C'[u_i] = \emptyset$ because $\Gamma \subseteq N(u_i)$ and $\forall u_j \in \Gamma, M[u_j] = \hat{M}[u_j]$.
This contradicts $M[u_i] \in C'_M[u_i]$.
Therefore, $\Gamma$ holds Eq. \ref{eq:pivot_set}.
In other words, $\Gamma$ is a dead-end.
\end{proof}

\subsubsection{Case 2: Non-injective Mapping}

For ease of discussion, we denote extended partial embedding $\hat{M} \cup \{(u_{k+1}, v)\}$ at line 9 by $\hat{M}_+$.
The if-statement at line 8 prevents the recursive call with non-injective $\hat{M}_+$, whose dead-end mask is given by the following lemma.

\begin{lemma}[Dead-end mask for the `non-injective mapping' case]
\label{lem:pivot_noninjective}
Let $\hat{M}_+$ be a partial embedding that violates the injection constraint.
Due to the violation, there exist $u_{i_1}$ and $u_{i_2}$ such that $\hat{M}_+[u_{i_1}] = \hat{M}_+[u_{i_2}]$.
Then, $\Gamma_+ = \{u_{i_1}, u_{i_2}\}$ is a dead-end mask of $\hat{M}_+$.
\end{lemma}
\begin{proof}
Arbitrary partial embedding $M$ violates the injection constraint if it contains $\{\, (u_i, v) \in \hat{M}_+[u_i] \,|\, u_i \in \Gamma_+ \,\}$.
Additionally, $\Gamma_+ \subseteq \mathrm{dom}(\hat{M}_+)$ holds.
Therefore, $\Gamma_+$ is a dead-end mask of $\hat{M}_+$.
\end{proof}

Notice that Lemma \ref{lem:pivot_noninjective} gives a dead-end mask of $\hat{M}_+$ while Lemma \ref{lem:pivot_no_candidate} gives a dead-end mask of $\hat{M}$.
We later describe how to convert a dead-end mask of $\hat{M}_+$ to that of $\hat{M}$.

\subsubsection{Case 3: Dead-end Pruning}

$\hat{M}_+$ is pruned if it matches dead-end pattern $D \in \mathcal{D}$.
In this case, we can obtain a dead-end mask from $D$.

\begin{lemma}[Dead-end mask for the `dead-end pruning' case]
\label{lem:pivot_deadend_match}
Let $D$ be a dead-end pattern.
$\Gamma_+ = \mathrm{dom}(D)$ is a dead-end mask of $\hat{M}_+$ if $D \subseteq \hat{M}_+$.
\end{lemma}
\begin{proof}
$\Gamma_+ \subseteq \mathrm{dom}(\hat{M}_+)$ holds because $D \subseteq \hat{M}_+$.
In addition, $\mathrm{Dead}(D)$ holds because $D$ is a dead-end.
Thus, $\Gamma_+$ is a dead-end mask of $\hat{M}_+$.
\end{proof}

\subsubsection{Case 4: Failing Recursion}
\label{sec:recursion_failed}

Algorithm \ref{alg:backtracking} recursively calls \textproc{Search} with $\hat{M}_+$.
If $\hat{M}_+$ is a dead-end, the callee is responsible for extracting a dead-end pattern.
By Lemma \ref{lem:pivot_deadend_match}, the caller can obtain a dead-end mask of $\hat{M}_+$ from the dead-end pattern extracted in the callee.

\subsubsection{Dead-end Mask Aggregation}

To extract a dead-end pattern from $\hat{M}$, function \textproc{Search} needs a dead-end mask of $\hat{M}$.
It is a subset of $\{u_1, u_2, \ldots, u_k\}$ by Definition \ref{def:deadend_mask}.
However, Lemma \ref{lem:pivot_noninjective} and \ref{lem:pivot_deadend_match} give a dead-end mask of $\hat{M}_+$, which may contain $u_{k+1}$ because $\hat{M}_+$ has a mapping of $u_{k+1}$.
Thus, we need to convert a dead-end mask of $\hat{M}_+$ to that of $\hat{M}$.

Lemma \ref{lem:pivot_noninjective} and \ref{lem:pivot_deadend_match} is used in the loop over $v \in C'[u_{k+1}]$ (line 7), and so we have dead-end masks for each $\hat{M}_+ = \hat{M} \cup \{(u_{k+1}, v)\}$.
From them, we can compute $\Gamma_*$, a set of query vertices, defined as follows:
\begin{equation}
\Gamma_* = \bigcup_{v \in C'[u_{k+1}]} \left( \text{dead-end mask of}\ \hat{M} \cup \{(u_{k+1}, v)\} \right).
\end{equation}
Here, the following lemma holds for $\Gamma_*$.

\begin{lemma}[Dead-end mask aggregation]
\label{lem:aggregated_deadend}
If $C'[u_{k+1}] \neq \emptyset$ holds, $\Gamma$ defined as follows is a dead-end mask of $\hat{M}$:
\begin{equation}
  \label{eq:aggregated_deadend}
  \Gamma = \left\{\begin{array}{ll}
    \left( \Gamma_* \cup N(u_{k+1}) \right) \cap \mathrm{dom}(\hat{M}) & \text{if}\ u_{k+1} \in \Gamma_* \\
    \Gamma_* & \text{if}\ u_{k+1} \not\in \Gamma_*.
  \end{array}\right. \\
\end{equation}
\end{lemma}
\begin{proof}
Since the case of $u_{k + 1} \not\in \Gamma_*$ is trivial, we consider the case of $u_{k + 1} \in \Gamma_*$.
Letting $\Gamma' = \Gamma_* \cap \mathrm{dom}(\hat{M})$, $D = \{(u_i, v) \in \hat{M} \,|\, u_i \in \Gamma' \}$, and $v \in C'[u_{k+1}]$, $D \cup \{(u_{k+1}, v)\}$ is a dead-end because it is a superset of the dead-end pattern extracted from $\hat{M} \cup \{(u_{k+1}, v)\}$.
Hence, $\mathfrak{D}_+ = \{ D \cup \{(u_{k+1}, v)\} \,|\, v \in C'[u_{k+1}] \}$ is a set of dead-end patterns.
Here, if we assume $\Gamma$ is not a dead-end mask of $\hat{M}$, there exists complete embedding $M$ such that $\{ (u_i, v) \in \hat{M} \,|\, u_i \in \Gamma \} \subseteq M$.
Letting $C'_M[u_{k+1}]$ be a candidate set obtained by substituting $M$ for $\hat{M}$ in Eq.\ \ref{eq:candidate_refinement}, $C'_M[u_{k+1}] = C'[u_{k+1}]$ since $N(u_{k+1}) \cap \mathrm{dom}(\hat{M}) \subseteq \Gamma$.
Thus, $\exists D_+ \in \mathfrak{D}_+,\, D_+ \subseteq M$ because $M[u_i] = \hat{M}[u_i]$ for $u_i \in \Gamma'$ and $M[u_{k+1}] \in C'[u_{k+1}]$.
However, this means $M$ is a complete embedding that contains a dead-end pattern.
By contradiction, $\Gamma$ is a dead-end mask of $\hat{M}$.
\end{proof}

Thus, we obtain a dead-end mask $\Gamma$ of $\hat{M}$.
Function \textproc{Search} use this for extracting a dead-end pattern from $\hat{M}$.

\subsection{Management of Daed-end Patterns}
\label{sec:deadend_management}

For simplicity, Section \ref{sec:pruning_by_deadend} described that our method records dead-end patterns in set $\mathcal{D}$ and prune partial embeddings if they match dead-end patterns in $\mathcal{D}$.
However, this mechanism is impractical if it is straightforwardly implemented because of spatial and temporal limitations.
From the spatial aspect, the number of dead-end patterns may increase up to the number of all the possible combinations of candidate vertices (i.e., $|C[u_1]| |C[u_2]| \ldots |C[u_n]|$).
Thus, the size of $\mathcal{D}$ may exceed the memory capacity.
From the temporal aspect, the pruning incurs overheads to check if the partial embedding contains a dead-end pattern in $\mathcal{D}$.
If we employ linear search over $\mathcal{D}$ and element-wise set containment tests, the overheads will be unacceptably large.
For making our pruning method feasible, this section introduces two techniques to efficiently manage dead-end patterns.

\subsubsection{Dead-end Patterns in Fixed-size Hash Table}

To mitigate the spatial and temporal issues, we employ a hash table to store dead-end patterns.
The key in the table is a vertex mapping (i.e., a pair of a query vertex and data vertex) added to the partial embedding at last.
Specifically, letting $\Delta$ be a hash table and $\hat{M} = \{(u_1, v_{i_1}), (u_2, v_{i_2}), \ldots, (u_k, v_{i_k})\}$ be a dead-end partial embedding, a dead-end pattern extracted from $\hat{M}$ is stored at $\Delta[u_k, v_{i_k}]$.
This is because the last mapping is most rarely overwritten.
We can also look up a dead-end pattern that may match the partial embedding by using its last mapping as a key.
This can be done in $O(1)$ time.

\subsubsection{Numeric Representation of Dead-end Patterns}
\label{sec:numeric_representation}

Storing dead-end patterns in the hash table offers efficient look-ups of dead-end patterns.
However, since dead-end pattern $D$ contains up to $|V_Q|$ mappings, it requires $O(|V_Q|)$ time to check if partial embedding $\hat{M}$ matches $D$ (i.e., $D \subseteq \hat{M}$).
This matching check is performed frequently, and thus largely affects the performance.

To address this problem, our method represents a dead-end pattern with a single integer.
This enables the matching check in $O(1)$ time.
Due to the space limitation, we here describe the basic idea only.
Our idea is based on the property of recursive calls in the backtracking.
Specifically, it exploits that function \textproc{Search} is called with unique partial embedding $\hat{M}$ for each call.
With this property, we use the call count of \textproc{Search} as an ID number (embedding ID) of partial embeddings.
For example, letting $\nu(\hat{M})$ be a function that maps partial embedding $\hat{M}$ to its embedding ID, 
we have $\nu(\{(u_1, v_1)\}) = 1$, $\nu(\{(u_1, v_1), (u_2, v_2)\}) = 2$, and $\nu(\{(u_1, v_1), (u_2, v_2), (u_3, v_5)\}) = 3$ in the search tree shown in Fig.\ \ref{fig:injection_example_searchtree_gup}.

However, there are still two problems in the numeric representation of dead-end patterns.
First, a partial embedding has an embedding ID only if it contains mappings of sequential query vertices (i.e., $u_1, u_2, u_3, \ldots$) although a dead-end pattern may lack some of them.
To mitigate this problem, we exploit that dead-end patterns are stored in the hash table.
The hash table $\Delta$ always stores a dead-end pattern extracted from $\hat{M}$ at $\Delta[u_k, \hat{M}[u_k]]$ where $k = |\hat{M}|$.
Since its location tells the last mapping of the dead-end pattern, we can ignore the last mapping in conversion from the dead-end pattern to an embedding ID.
Specifically, when we extract a dead-end pattern from $\hat{M}$ using its dead-end mask $\Gamma$, we ignore $u_k$ in $\Gamma$ (if it exists).
Second, we need to manage mappings from a partial embedding to an embedding ID.
The number of embedding IDs equals the number of calls of \textproc{Search}, and so they cannot be on memory.
To solve this, we manage embedding IDs only of every subset of the current partial embedding.
When the algorithm is processing $\hat{M}$, we maintain $\Phi$, an array of embedding IDs, so that $\Phi[\mu] = \nu(\{ (u_i, v) \in \hat{M} \,|\, i \leq \mu \})$.
This can be done by simply recording the number of calls in $\Phi[u_k]$ when $\textproc{Search}(\hat{M}, C)$ is called.

By using these techniques, given a dead-end partial embedding $\hat{M}$ and its dead-end mask $\Gamma$, the dead-end pattern can be represented by triplet of (i) an embedding ID, (ii) the number of mappings in a partial embedding that the embedding ID represents, and (iii) an original dead-end mask.
Specifically, we store dead-end patterns as follows:
\begin{equation}
\Delta[u_k, \hat{M}[u_k]] \gets \left( \Phi[\mu], \mu, \Gamma \right).
\end{equation}
where $\mu = \max_{u_i \in \Gamma, i < k} i$.
The triplet has dead-end mask $\Gamma$ to obtain a new dead-end pattern by using Lemma \ref{lem:pivot_deadend_match}.
Given partial embedding $\hat{M}$ ($k = |\hat{M}|$) and array of its embedding IDs $\Phi$, we can check if $\hat{M}$ matches a dead-end pattern by the following condition:
\begin{equation}
\Phi[\mu] = \phi, \ \text{where}\ (\phi, \mu, \Gamma) = \Delta[u_k, \hat{M}[u_k]].
\end{equation}
The whole matching check can be performed in $O(1)$ time because both access to $\Delta$ and comparison of embedding IDs consume $O(1)$ time.
Due to space limitation, the algorithm shown in the next subsection omits details of the numerical representation of dead-end patterns.
It internally uses representations described above to reduce overheads of dead-end pruning.

\subsection{Algorithm Details}

Algorithm \ref{alg:search_details} shows our subgraph matching algorithm.
Note that this algorithm assumes that global variable $\Delta$ is a hash table of dead-end patterns.
Function \textproc{Search} takes partial embedding $\hat{M}$ and candidate set $C$ and returns a dead-end mask of $\hat{M}$ if $\hat{M}$ is found to be a dead-end; otherwise, it returns an empty set.
Function call $\textproc{Search}(\emptyset, C)$ starts the search.
Algorithm \ref{alg:search_details} differs from Algorithm \ref{alg:backtracking} in dead-end mask selection (e.g.\ line 8), pruning (line 14), and dead-end pattern recording (line 20).
Dead-end mask selection is performed together with checks for each reason of search failures.
Except for the `empty candidate set' case, Lemmas give a dead-end mask of the extended partial embedding $M \cup \{(u_{k+1}, v)\}$.
We accumulate these masks in $\Gamma_*$ and convert it to the dead-end mask of $\hat{M}$ after the loop (line 18).
The second difference, matching with dead-end patterns, is lightweight because the partial embedding is compared with only one element in the hash table, i.e., $\Delta[u_{k+1}, v]$.
Note that we assume that $\Delta[u_{k+1}, v] \subseteq \hat{M}$ is false if $\Delta[u_{k+1}, v]$ is not yet defined during the algorithm.
The third difference, dead-end pattern recording, is performed if the recursively called function does not report any complete embeddings.
This condition equals that $\hat{M}$ is a dead-end and $\hat{M}$ is not empty (line 19).
The latter condition is checked to avoid accesss to undefind value $\hat{M}[u_k]$ if $\hat{M} = \emptyset$.
Last, the algorithm returns a dead-end mask if $\hat{M}$ is a dead-end (line 21); otherwise, it returns an empty set (line 22).

\begin{algorithm}[!t]
\caption{Detailed Search Algorithm}
\begin{algorithmic}[1]
  \Require Query graph $Q$ and data graph $G$
  \Ensure All the embeddings of $Q$ in $G$

  \Function{Search}{$\hat{M}, C$}
    \State $k \leftarrow |\hat{M}|$
    \If{$k = |V_Q|$}
      \State Report $\hat{M}$ as a complete embedding
      \State \Return $\emptyset$
    \EndIf
    \State $C' \leftarrow$ Candidates refined with the edge constraints (Eq.\ \ref{eq:candidate_refinement})
    \If{There exists $u_i$ s.t.\ $C'[u_i] = \emptyset$}
      \State $\Gamma \leftarrow N(u_i) \cap \mathrm{dom}(\hat{M})$  \Comment{Lemma \ref{lem:pivot_no_candidate}}
    \Else
      \State $\Gamma_* \leftarrow \emptyset$
      \ForAll{$v \in C'[u_{k+1}]$}
	\If{$v \in \mathrm{dom}(\hat{M})$}
	  \State $\Gamma_* \leftarrow \Gamma_* \cup \{\, u_i \in \mathrm{dom}(\hat{M}) \,|\, \hat{M}[u_i] = v \,\}$  \Comment{Lemma \ref{lem:pivot_noninjective}}
	\ElsIf{$\Delta[u_{k+1}, v] \subseteq \hat{M}$}
	  \State $\Gamma_* \leftarrow \Gamma_* \cup \Delta[u_{k+1}, v]$  \Comment{Lemma \ref{lem:pivot_deadend_match}}
	\Else
	  \State $\Gamma_* \leftarrow \Gamma_* \cup \textproc{Search}(\hat{M} \cup \{(u_{k+1}, v)\}, C)$
	\EndIf
      \EndFor
      \State $\Gamma \leftarrow \Gamma_*$ converted with Eq.\ \ref{eq:aggregated_deadend}  \Comment{Lemma \ref{lem:aggregated_deadend}}
    \EndIf
    \If{No report in the recursive call and $\hat{M} \neq \emptyset$}
      \State $\Delta[u_k, \hat{M}[u_k]] \leftarrow \{\, (u_i, v) \in \hat{M} \,|\, u_i \in \Gamma \,\}$
      \State \Return $\Gamma$
    \EndIf
    \State \Return $\emptyset$
  \EndFunction
\end{algorithmic}
\label{alg:search_details}
\end{algorithm}

This algorithm holds the following property.

\begin{theorem}[Completeness]
Algorithm \ref{alg:search_details} reports all the embeddings of query graph $Q$ within data graph $G$.
\end{theorem}
\begin{proof}
Compared with the naive backtracking (Algorithm \ref{alg:backtracking}), this algorithm searches differently because of matching with dead-end patterms at line 14.
From Lemma \ref{lem:pivot_no_candidate}, \ref{lem:pivot_noninjective}, \ref{lem:pivot_deadend_match}, and \ref{lem:aggregated_deadend}, $\Delta[u_{k+1}, v]$, partial embedding $\hat{M}$ is a dead-end if it holds $\Delta[u_{k+1}, v] \subseteq \hat{M}$.
From Definition \ref{def:deadend}, dead-end partial embeddings do not yield a complete embedding.
Therefore, this algorithm reports all the complete embeddings even if it prune dead-end partial embeddings.
\end{proof}

\section{Evaluation}
\label{sec:evaluation}

This section evalutes the performance of our algorithm.
We implemented our algorithm with structural analysis-based pruning and matching order selection that are proposed in CFL-Match, the state-of-the-art method \cite{Bi2016}.
The evaluation also uses CFL-Match, QuickSI \cite{Shang2008}, and GraphQL\cite{He2008} for comparison.
Quick SI and Graph QL shows high performance in \cite{Lee2012a}.
We obtained an implementation of CFL-Match from its author and implementations of Quick SI and Graph QL from the author of \cite{Lee2012a}.
Our machine is equipped with Intel Xeon E5-2697 v2 and 128GB memory.
We use yeast and human as a data graph, which are widely used in the previous studies \cite{He2008, Shang2008, Zhao2010, Han2013, Ren2015, Bi2016}.
Both are a protein-protein interaction network.
yeast has 3112 vertices, 12519 edges, and 71 vertex labels.
human has 4674 vertices, 86282 edges, and 44 vertex labels.
We generate query graphs by extracting a connected component in a data graph with a random walk.
In the experiments, each algorithm processes a query set that contains 10000 queries.
Query sets vary in the number of vertices in a query graph.
If an algorithm cannot process a query set in one day, we consider it Do-Not-Finish (DNF).
Since it may be impractical to enumerate all the embeddings due to the combinatorial explosion, we stop enumerating if 1000 embeddings are found, similarly to \cite{Han2013}.

\subsection{Query processing time}
\label{sec:wla_runtime}

\begin{figure}[t]
  \centering
  \subfloat[yeast data set]{
    \includegraphics[width=0.48\hsize, page=1]{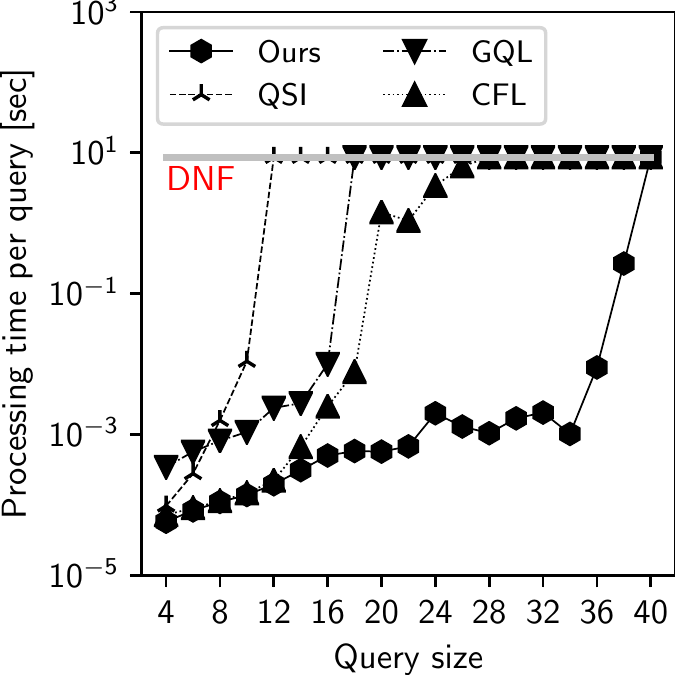}
    \label{fig:wla_yeast_path}
  }
  \subfloat[human data set]{
    \includegraphics[width=0.48\hsize, page=1]{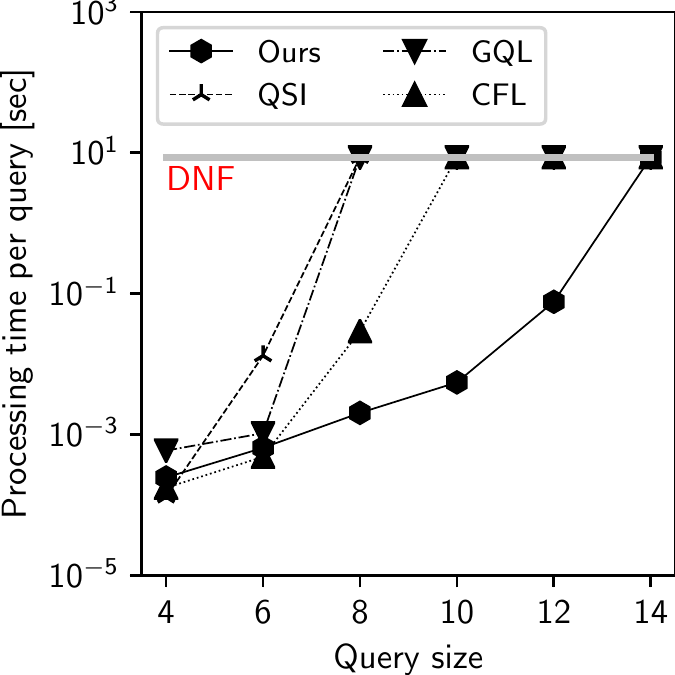}
    \label{fig:wla_yeast_tree}
  }
  \caption{Average procesing time for one query of each algorithm}
  \label{fig:wla_yeast}
\end{figure}

First, we compare query processing time to evaluate the performance improvement yielded by dead-end pruning.
Fig.\ \ref{fig:wla_yeast} shows the result.
Ours, QSI, GQL, and CFL represent our method, QuickSI, GraphQL, and CFL-Match, respectively.
All the methods become DNF for 40-vertex queries on yeast and 14-vertex queries on human.
As a whole, our method shows the best performance for almost all the query size and data sets.
It is especially efficient for large-scale queries.
For example, it shows 1000 times higher performance for 26--36-vertex queries on yeast.
As shown in Fig.\ \ref{fig:injection_example}, the effectiveness of the exisitng methods sensitively depend on structures of given graphs.
This matters more seriously for larger queries because they have the more combinations of candidate vertices.
Compared with the existing methods, our method prunes unnecessary searches by dead-end pruning.
This reduces the sensitiveness to graph structures.
Thus, our method can reduce processsing time even for large queries.

\subsection{The number of pruning}

\begin{figure}[t]
  \centering
  \subfloat[yeast dataset]{
    \includegraphics[width=0.48\hsize, page=1]{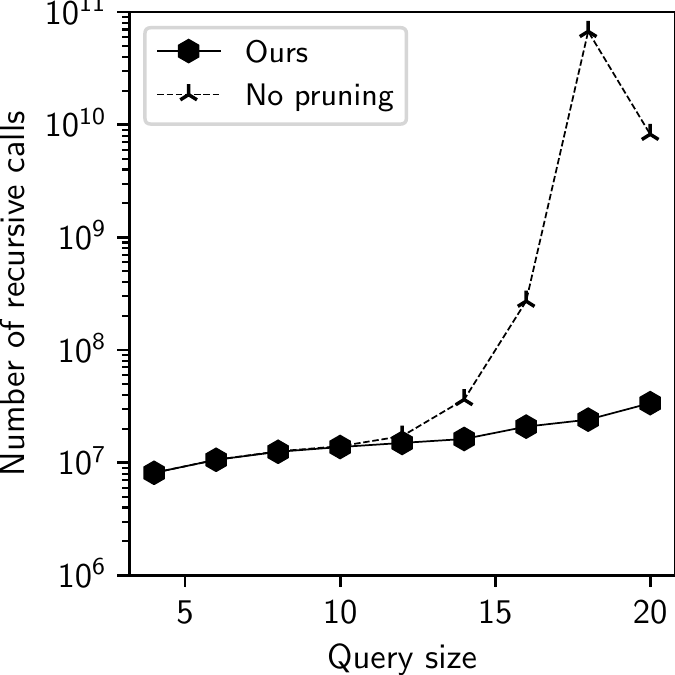}
    \label{fig:rec_yeast_walk}
  }
  \subfloat[human dataset]{
    \includegraphics[width=0.48\hsize, page=1]{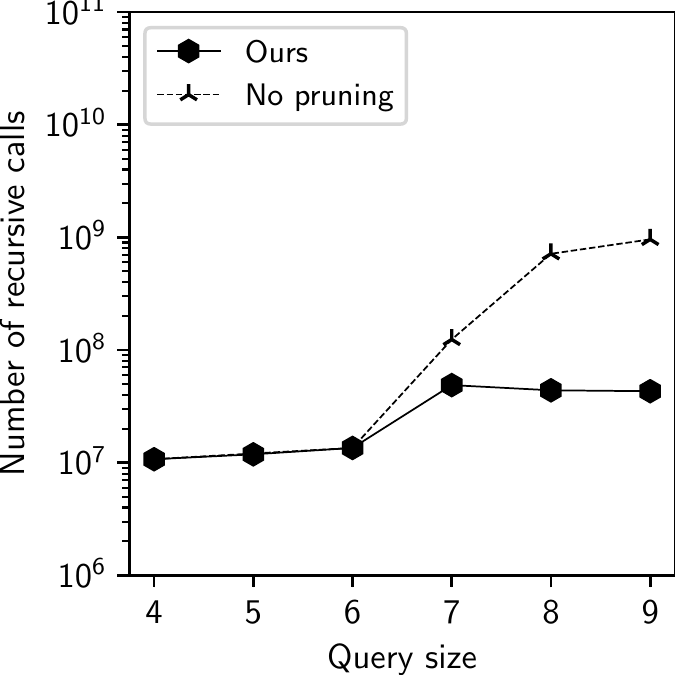}
    \label{fig:rec_human_walk}
  }
  \caption{The number of recursive calls during processing one query set}
  \label{fig:rec_walk}
\end{figure}

Next, to understand what offers the performance improvement, we focus on the number of recursive calls of function \textproc{Search} (Algorithm \ref{alg:search_details}).
This experiment compares the number of calls between our method (`Ours') and our method without dead-end pruning (`No pruning') that is identical to `Ours' except that it lacks lines 14 and 15 of Algorithm \ref{alg:search_details}.
The difference in the number of recursive calls shows an effect of dead-end pruning.
Fig.\ \ref{fig:rec_walk} shows the result.
We omit the results for queries which have over 20 vertices on yeast and over 9 vertices on human because those cause DNF for `No pruning'.
The number of recursion shows that there are few prunings for small queries, but it increases as the size of queries increase.
For example, `No pruning' recurses about $6.7 \times 10^{10}$ times for 18-vertex queries on yeast, but `Ours' recurses only about $2.4 \times 10^7$ times.
This is because larger queries tend to involve more search failures.
The number of violations of the injection constraint and the edge constraint increases along with the number of vertices and edges in the query graph.
Our method significantly improves the performance by reducing search failures caused by these reasons.
Our method also shows the comparable performance for small queries.
This is because the overheads for dead-end pruning is small owing to the effcient management of dead-end patterns described in Section \ref{sec:deadend_management}.

\section{Conclusion}
\label{sec:conclusion}

Subgraph matching is widely used, but it suffers from high computational cost due to its NP-hardness.
This paper propose a subgraph matching algorithm that improves the performance by learning from failures.
Specifically, it generates dead-end patterns from partial embeddings that caused a search failure during the backtracking and, in the subsequent process, prunes partial embeddings that match dead-end patterns.
The experimental results show that our method is up to 10000 times faster than existing methods.

% References and Appendix  <<<

% ensure same length columns on last page (might need two sub-sequent latex runs)
\balance

%ACKNOWLEDGMENTS are optional
%\section{Acknowledgments}

\bibliographystyle{abbrv}
\bibliography{library}  % vldb_sample.bib is the name of the Bibliography in this case

%APPENDIX is optional.
% ****************** APPENDIX **************************************
% Example of an appendix; typically would start on a new page
%pagebreak

%\begin{appendix}
%You can use an appendix for optional proofs or details of your evaluation which are not absolutely necessary to the core understanding of your paper. 
%
%\section{Final Thoughts on Good Layout}
%Please use readable font sizes in the figures and graphs. Avoid tempering with the correct border values, and the spacing (and format) of both text and captions of the PVLDB format (e.g. captions are bold).
%
%At the end, please check for an overall pleasant layout, e.g. by ensuring a readable and logical positioning of any floating figures and tables. Please also check for any line overflows, which are only allowed in extraordinary circumstances (such as wide formulas or URLs where a line wrap would be counterintuitive).
%
%Use the \texttt{balance} package together with a \texttt{\char'134 balance} command at the end of your document to ensure that the last page has balanced (i.e. same length) columns.
%
%\end{appendix}

% >>>
\end{document}